\documentclass[submission,copyright,creativecommons]{eptcs}
\providecommand{\event}{AFL 2023} 
\usepackage{tikz}

\usepackage{iftex}
\usepackage{caption}
\usepackage{subcaption}
\ifpdf
  \usepackage{underscore}         
  \usepackage[T1]{fontenc}        
\else
  \usepackage{breakurl}           
\fi
\usepackage{graphicx}
\usepackage[english]{babel}
\usepackage{amsmath,amssymb,amsfonts}
\usepackage{xspace}
\usepackage{comment}

\newcommand{\SV}[1]{#1}
\newcommand{\LV}[1]{}

\usepackage{hyperref}

\usepackage{tikz,tikz-cd}

\DeclareSymbolFont{Shuffle}{U}{shuffle}{m}{n}
\DeclareFontFamily{U}{shuffle}{}
\DeclareFontShape{U}{shuffle}{m}{n}{%
  <-8>shuffle7%
  <8->shuffle10%
}{}
\DeclareMathSymbol\shuffle{\mathbin}{Shuffle}{"001}
\DeclareMathSymbol\cshuffle{\mathbin}{Shuffle}{"002}

\usepackage{pgf,tikz}
\usepackage{rotating}

\def \no {\noindent}

\def \l {\lambda}
\def \a {\alpha}
\def \b {\beta}

\def \ra {\rightarrow}
\def \Ra {\Rightarrow}

\def \beq{\begin{eqnarray*}}
\def \eeq{\end{eqnarray*}}

\newcommand{\emptyword}{\lambda}

\newcommand{\qed}{\hfill$\Box$}
\newenvironment{proof}{\underline{Proof}}{\qed}
 \newtheorem{theorem}{Theorem}
 \newtheorem{definition}[theorem]{Definition}
 \newtheorem{proposition}[theorem]{Proposition}
 \newtheorem{example}[theorem]{Example}
 \newtheorem{remark}[theorem]{Remark}

\title{When Stars Control a Grammar's Work}
\def\titlerunning{When Stars Control a Grammar's Work}
\def\authorrunning{H. Fernau, L. Kuppusamy, and I. Raman}

\author{Henning Fernau
\institute{Abteilung Informatikwissenschaften\\  
Universit\"at Trier, 
Germany} 
\email{fernau@uni-trier.de}
\and 
Lakshmanan Kuppusamy
\institute{School of 
Computer Science and 
Engineering\\ VIT, 
Vellore
, India} \email{klakshma@vit.ac.in}
\and 
Indhumathi Raman
\institute{Department of Computing Technologies, \\
School of Computing, SRMIST, Chennai, India} \email{indhumar2@srmist.edu.in}
}

\begin{document}
\maketitle

\begin{abstract}
Graph-controlled insertion-deletion (GCID) systems are regulated extensions of insertion-deletion systems. Such a system has several components and each component contains some insertion-deletion rules. The components are the vertices of a directed control graph\LV{, whose arcs describe how strings can move between components}.
A rule is applied to a string in a component and the resultant string is moved to the target component specified in the rule.
The language of the system is the set of all terminal strings collected in the final component.
\LV{In this paper, w}\SV{W}e 
impose the restriction in the structure of the underlying graph to be a star structure where there is a central, control component which acts like a master and transmits a string (after applying one of its rules) to one of the components specified in the (applied) rule. A component which receives the string can process the obtained string with any applicable rule available in it and sends back the resultant string only to the center component. 
With this restriction, we obtain computational completeness for some descriptional complexity measures. 
\LV{The obtained results can also be applicable to membrane computing as the graph structure of the system is a star.}


\end{abstract}


\section{Introduction}

%
Insertion-deletion systems are part of formal languages which are extensively analyzed. The motivation for the systems comes from both linguistics \cite{Mar69,Pau97} and molecular biology. The action of inserting or deleting some strands do occur often in 
DNA processing\SV{~\cite{PauRozSal98a}} and RNA editing\SV{~\cite{Ben93}}. 
\LV{Besides, a mathematical motivation for insertion operation can be found in \cite{Hau83}, whereas the deletion operation was first studied in \cite{Kar91}. }These two operations together were 
introduced as a  
formal languages theory framework in
\cite{KarThi96} and further studied in \cite{KarPTY99,TakYok2003}. The corresponding 
grammatical mechanism is called {\it insertion-deletion system} (abbreviated as ins-del system). \SV{The insertion operation means inserting a string $\eta$ in between the strings $w_1$ and $w_2$, whereas the deletion operation is deleting a substring $\delta$ from the string $w_1\delta w_2$.}\LV{ Informally, 
if a string $\eta$ is inserted
between two parts $w_1$ and $w_2$ of a string $w_1w_2$ to get $w_1 \eta w_2$,
we call the operation {\it insertion}, whereas if a substring $\delta$ is
deleted from a string $w_1 \delta w_2$ to get $w_1w_2$, we call the operation {\it deletion}.} 

In the literature, several variants of ins-del systems have been considered\SV{. We refer to the survey article~\cite{Ver2010} for details concerning the state-of-the-art around 2010.}\LV{, like ins-del P systems \cite{AlhKRV2011,KriRam0102}, tissue P systems with ins-del rules~\cite{KupRam2003}, context-free ins-del systems~\cite{MarPRV2005}, 
matrix ins-del systems \cite{KupMahKri2011,PetVer2012,FerKupRam2018,FerKupRam2021a}, random context and semi-conditional ins-del systems~\cite{IvaVer2015a} 
etc. All the mentioned papers (as well as~\cite{KarPTY99,TakYok2003}) 
characterized the recursively enumerable languages 
(\emph{i.e.}, \emph{computational completeness}) using ins-del systems.
We refer to the survey article~\cite{Ver2010} for details of variants of ins-del systems; this survey also discusses some proof techniques for showing computationally completeness results.} 
One of the important variants of ins-del systems is {\it graph-controlled ins-del systems} (abbreviated as GCID systems), introduced in \cite{FreKRV2010} and further studied in \cite{IvaVer1314}. In such a  system, the concept of \emph{components} is introduced, which are associated with insertion or deletion rules. The transition is performed by choosing any applicable rule from the set of rules of the current component and by moving the resultant string to the target component specified in the rule in order to continue processing it. 
Several restrictions of graph control have been studied, e.g., matrix ins-del systems (see \cite{PetVer2012,FerKupRam2021a} and more papers cited there), time-varying ins-del systems~\cite{AlhFIV2022},  or path-controlled ins-del systems~\cite{FerKupRam2019}.
In this paper, we consider star control (which also has been considered in~\cite{IvaVer2017} in an implicit way when dealing with graph-controlled systems with two components.
This models a kind of master-slave system in the sense that the central component always dispatches work to the slave components who, after finishing their work, return the result to the master component. Graph-controlled insertion-deletion systems whose underlying control graph is a tree are equivalent to ins-del P systems~\cite{KriRam0102,IvaVer1314}. Hence, our restriction can be viewed as a special case of  ins-del P systems. We want to point to one technicality here: with P systems (and similarly with several restrictions of graph control), there is the possibility to stay in the same membrane with the \emph{here} command (which corresponds to allowing loops in graph control); in the model that we consider in this paper, this is disallowed: when the master saw the `current work' (and worked on it one step), it has to pass it to some slave immediately, and after the slave performed one step, the work is handed back to the master, etc. Therefore, results in the literature concerning seemingly related models do not always compare well to this star model.

The \emph{descriptional complexity} of a GCID system is measured by its \emph{size} $s=(k;n,i',i'';m,j',j'')$, where the parameters represent resource bounds as given in Table~\ref{table-parameters}. Slightly abusing notation, the language class generated by GCID systems of size $s$ is denoted by $\mathrm{GCID}(s)$. We attach subscripts $P$ and $S$ when referring to path-controlled or star-controlled GCID systems, respectively.

\begin{table}[t]
\begin{center}
$$\begin{array}{|lcl|lcl|}
\hline 
k &=& \textrm{the~number~of~components}&&& \\
n &=& \max \{|\eta | \colon (i,(u, \eta, v)_I,j) \in R\} &
m&=& \max \{|\delta | \colon (i,(u, \delta, v)_D,j) \in R\}\\
i'&=& \max \{|u| \colon (i,(u, \eta, v)_I,j) \in R \} &
j'&=& \max \{|u| \colon (i,(u, \delta, v)_D,j) \in R \} \\
i''&=&\max \{|v| \colon (i,(u, \eta, v)_I ,j) \in R \} & 
j''&=&max \{|v| \colon (i,(u, \delta, v)_D,j) \in R \} \\
\hline 
\end{array}$$
\caption{Size $(k;n,i',i'';m,j',j'')$ of a GCID system}
\label{table-parameters}
\end{center}
\end{table}

\noindent
The main results of this paper are the following ones.

\begin{enumerate}


\item (Theorem \ref{RE6110200}) $\mathrm{RE}=\mathrm{GCID}_S(6;1,1,0;2,0,0) = \mathrm{GCID}_S(6;1,0,1;2,0,0) $; 

\item (Theorem \ref{RE4211100}) $\mathrm{RE}=\mathrm{GCID}_S(4;2,1,1;1,0,0)$.

\end{enumerate}

Our proofs are based on the {\it Special Geffert Normal Form} of type-0 grammars, which characterizes the class $\mathrm{RE}$, the recursively enumerable languages. Formal definitions follow in the next section.

\section{Preliminaries}\label{sec-prelim}
We assume that the readers are familiar with the standard notations used in formal language theory. 
Here, we recall a few notations for the understanding of the paper.
  Let $\mathbb{N}$ denote the set of positive integers, and $[\ell\ldots k]=\{i\in\mathbb{N}\colon \ell\leq i\leq k\}$. 
Given an {\it alphabet} (finite set)  $\Sigma$, $\Sigma^*$ denotes the free monoid generated by $\Sigma$. The elements of $\Sigma^*$ are called {\it strings} or {\it words};  $\emptyword$ denotes the empty string. For a string $w \in\Sigma^*$, $|w|$ is the length of $w$ and $w^R$ denotes the reversal (mirror image) of $w$. 
$L^R$ and $\mathcal{L}^R$ are also understood for languages $L$ and language families $\mathcal{L}$, collecting all reversals of words from $L$ and all reversals of languages from $\mathcal{L}$, respectively. 
  For the computational completeness results, as our main tool we  use  the fact that type-$0$ grammars\LV{\footnote{A type-$0$ grammar $G$ is usually specified by a quadruple $(N,T,P,S)$ consisting of a nonterminal alphabet $N$, a terminal alphabet $T$, a finite set of (production) rules $P$ and a start symbol $S\in N$. Rules are written in the form $\a\to\b$, $\a,\b\in(N\cup T)^*$. This defines a rewrite relation $\Ra_G\subseteq (N\cup T)^*\times (N\cup T)^*$, with $u\Ra_Gv$ if $v$ is obtained from~$u$ by replacing the subword $\a$ by~$\b$, for some  $\a\to\b\in P$. The reflexive transitive closure $\Ra_G^*$ can be used to define the semantics of $G$ ---the language of~$G$--- collecting 
  all $w\in T^*$ with $S\Ra_G^*w$.}} in Special Geffert Normal Form (SGNF) \LV{are known to characterize}\SV{describe} \SV{$\mathrm{RE}$}\LV{the recursively enumerable 
  languages}.
\begin{definition}\label{def-SGNF}
A type-$0$ grammar $G=(N,T,P,S)$ is said to be in SGNF if 
\begin{itemize}
\item 
$N$ decomposes as $N=N'\cup N''$, where $N''=\{A,B,C,D\}$ and $N'$ contains at least the two nonterminals $S$ and $S'$; 
\item the only non-context-free rules in $P$ are $AB \to\emptyword$ and $CD\to\emptyword$;
\item the context-free rules are of the form (i) $S'\to \emptyword$, or (ii) $X \to Y_1Y_2$, where $X\in N'\setminus\{S'\}$ and $Y_1Y_2\in ((N'\setminus\{X\})(T \cup \{B,D\})) \cup (\{A,C\}(N'\setminus\{X\}))$.

\end{itemize}
\end{definition}

The way the normal form is constructed is described in \cite{FreKRV2010}, based on \cite{Gef91a}. 
We can assume that $S'$ does not appear on the left-hand side of any non-erasing rule. This also means that the derivation in~$G$ undergoes two phases. In phase I, only context-free rules are applied. This phase ends with applying the context-free deletion rule $S' \ra \l$ (which is the only rule that has~$S'$ on its left-hand side). Then only, non-context-free deletion rules $AB \to \l$ and $CD \to \l$ are applied in phase II. Notice that the symbol from~$N'$, as long as present, separates $A$ and $C$ from $B$ and $D$;  this prevents a premature start of phase II. One of the features of SGNF derivations is that any string that can be derived can contain at most one substring $AB$ or $CD$ in its so-called \emph{center}. If such a substring is present, the derivation is in phase II; also, then no nonterminal from $N'$ occurs. In phase I, exactly one such nonterminal is present (in the center). Therefore, we can differentiate two cases within (ii) for $X,Y\in N'\setminus\{S'\}$ with $X\neq Y$: either, we have a rule $X\to bY$, with $b\in \{A,C\}$, or we have a rule $X\to Yb$, with $b\in T\cup\{B,D\}$.
This case distinction is often necessary when simulating type-0 grammars in SGNF, as we will see later.

We write $\Ra_r$ to denote a single derivation step using rule~$r$, and $\Ra_G$ (or $\Ra$ if no confusion arises) denotes a single derivation step using any rule of $G$. Then, $L(G)=\{w\in T^*\mid S\Ra^* w\}$, where $\Ra^*$ is the reflexive transitive closure of $\Ra$. 

\subsection{Graph-Controlled Insertion-Deletion Systems}

\begin{definition}
A \emph{graph-controlled insertion-deletion system} (GCID for short) with $k$ components
  is a construct $\Pi = (k,V,T,A,H,i_0,F,R)$, where 
\SV{(i) $k$ is the number of components, (ii) $V$ is an alphabet, (iii) $T \subseteq V$ is the terminal alphabet, (iv) $A \subset V^*$ is a finite set of strings, called {\it axioms}, 
present in the initial component,
 (v) $H$ is a set of labels associated (in a one-to-one manner) to the rules in $R$, (vi) $i_0 \in [1 \ldots k]$ is the initial component, (vii) $F\subseteq [1 \ldots k]$ is the set of final components and (viii) $R$ is a finite set of rules of the form $l: (i,r,j)$,  where $l\in H$ is the label of the rule, $r$ is an insertion rule of the form $(u, \eta, v)_I$, with insertion string~$\eta$ and context $(u,v)$, or deletion rule of the form $(u, \delta, v)_D$, with deletion string~$\delta$ and context $(u,v)$, where $u,v \in V^*, ~\eta,\delta \in V^+$ and $i, j \in [1 \ldots k]$.
} 
\end{definition}
\LV{ If one of the $u$ or $v$ is $\emptyword$ for all the insertion (deletion) contexts, then we call the insertion (deletion) as one-sided.
If both $u,v=\emptyword$ for every insertion (deletion) rule, then it means that the corresponding insertion (deletion) can be done freely anywhere in the string and is called \emph{context-free} insertion (context-free deletion).} 
Often, the component is part of the label name. This will also (implicitly) define $H$. 
We shall omit the label $l$ of the rule wherever it is not necessary for the discussion. 

We now describe how GCID systems work. 
Applying an insertion rule of the form $(u,\eta ,v)_I$ means that the string $\eta$ is inserted between $u$ and $v$; this corresponds to the rewriting rule
$uv \rightarrow u \eta v$. Similarly, applying a deletion rule of the form $(u,
\delta,v)_D$ means that the string $\delta$ is deleted between $u$
and $v$; this corresponds to the rewriting rule $u \delta  v \rightarrow uv$. 
A \emph{configuration} of~$\Pi$ is represented by $(w)_i$,  where $i\in [1\ldots k]$ is the number of the current component 
and $w\in V^*$ is the current string. \LV{In that case, w}\SV{W}e also say that $w$ has entered or moved to component $Ci$. 
We write $(w)_i \Rightarrow_{l} (w')_j$ 
if there is a rule $l: (i,r,j)$ in $R$, and $w'$ is obtained by applying the insertion or deletion rule~$r$ to $w$. 
For brevity, we write $(w)_i\Ra (w')_j$ if there is some rule $l$ in $R$ such that  $(w)_i\Ra_l (w')_j$.
To avoid confusion with traditional grammars, we write $\Ra_*$ for the transitive reflexive closure of $\Ra$ between configurations. The language $L(\Pi)$ generated by $\Pi$ is defined as $\{w\in T^*\mid (x)_{i_0}\Ra_* (w)_{i_f} ~\text{for some}~ x \in A$ and some $i_f\in F \}$.

The \emph{underlying control graph} of a \LV{graph-controlled insertion-deletion}\SV{GCID} system $\Pi$ with $k$ components is defined to be a graph with $k$ nodes labelled $C1$ through $Ck$ and there exists a directed edge from \LV{a node }$Ci$ to \LV{node }$Cj$ if there exists a rule of the form $(i,r,j)$ in $R$ of~$\Pi$. We also associate an undirected graph on $k$ nodes to a GCID system of $k$ components as follows: there is an undirected edge from a node $Ci$ to $Cj$ if there exists a rule of the form $(i,r_1,j)$ or $(j,r_2,i)$ in $R$ of~$\Pi$. 
We call a GCID system with $k$ components \emph{star-controlled} if its underlying undirected control graph has the edge set $\left\{\{C1, Ci\} \mid i\in [2\dots k] \right\}$.
\LV{Notice that t}\SV{T}his means that the corresponding directed control graph may contain arcs like $(C1,Ci)$, $(Ci,C1)$, but no loops\LV{ $(Cj,Cj)$}. 

Below, we provide a few examples for a better understanding of how the above-defined system works. As star-controlled systems have to have at least two components to produce anything non-trivial, it is interesting to observe that even with only two components, non-regular languages can be obtained.

\begin{example}
The language $\{a^nb^n \mid n \geq 0\}$ can be generated by a star-controlled insertion-deletion system 
with two components, alphabet $\{a,b,A,B\}$, the axiom set $\{AB\}$ in C1, final component $\{C1\}$ 
and the following rules:
$r1.1:(1,(A,a,\l)_I,2)$,
$r2.1:(2,(B,b,\l)_I,1)$,
$r1.2:(1,(\l,A,\l)_D,2)$ and
$r2.2:(2,(\l,B,\l)_D,1)$. 
A possible derivation of a terminal string is:
$$(AB)_1\Ra_{r1.1}(AaB)_2\Ra_{r2.1}(AaBb)_1\Ra_{r1.2}(aBb)_2\Ra_{r2.2}(ab)_1.$$
Observe that $(aBb)_2\Ra_{r2.1}(aBbb)_1$ is possible, but now the derivation is stuck, as any rule in $C1$ checks for the presence of the nonterminal~$A$. Yet, as the nonterminal~$B$ is present, the configuration  $(aBbb)_1$ cannot lead to a terminal word. The size of this system is $(2;1,1,0;1,0,0)$. A very similar system can be given for this language that is of size $(2;1,0,1;1,0,0)$. For the very similar language $\{a^nb^n\mid n\geq 1\}$, even a system with two rules $r1:(1,(a,a,\l),2)$ and $r2:(2,(b,b,\l),1)$ would suffice, with axiom $ab$.
\end{example}

Recall that the class REG of regular languages is the lowest class of the Chomsky hierarchy. It can be characterized by right-linear grammars whose rules have the form
$A\to Ba$ or $A\to\l$ for nonterminal symbols $A,B$ and a terminal symbol~$a$. We use this characterization to prove that star-controlled GCID systems can generate all regular languages. The previous example then shows that even non-regular languages can be generated.

\begin{theorem}Each regular language (and also some non-regular languages) can be generated by a $\mathrm{GCID}_S$ system of size $(2;3,0,1;2,0,0)$, where the initial component $C1$ is also the only final one. 
\end{theorem}
Later, we will see that with both components being final, many more languages can be described.

\noindent
\begin{proof}
We only sketch the construction in the following. For each rule of the form $A\to aB$ of a right-linear grammar~$G$ that we start with, we introduce the insertion rule $(\l,aB\$,A)$ into the first component of the  simulating  $\mathrm{GCID}_S$ system~$\Pi$.
For each erasing rule $A\to\l$, we add the insertion rule $(\l,a\$,A)$ into the first component of~$\Pi$.
In both cases, $\$$ is a special marker symbol that is taken care of in the second component that contains all possible deletion rules of the form $(\l,\$A,\l)$ for any nonterminal~$A$ of~$G$. For instance, if~$G$ contains the rules $S\to aX$ and $X\to\l$, enabling the derivation $S\Ra aX\Ra a$, then the simulation is performed as follows:
\LV{$$}\SV{$}(S)_1\Ra(aX\$ S)_2\Ra(aX)_1\Ra(a\$X)_2\Ra(a)_1\,.\LV{$$}\SV{$}
\end{proof}

By adding more nonterminal symbols, one can also achieve this result with $\mathrm{GCID}_S$ systems  of size $(2;2,0,1;2,0,0)$.
We leave it open if $\mathrm{GCID}_S$ systems with only two components and only one final component can generate each recursively enumerable language.

\begin{example}\label{exa-copy}
The copy language $\{ww \mid w \in \{a,b\}^* \}$ can be generated by a star-controlled insertion-deletion system 
$\Pi = (3;\{a,b,A,B\}, \{a,b\}, \{AB\}, H, 1, \{1\}, R)$, where $H=\{r1.1, r1.2, r1.3, r2.1, r2.2,\\ r3.1 \}$ and $R$ is the set of rules depicted in Table~\ref{tableCOPYlang}; $\Pi$ has size $(3;1,0,1;1,0,0)$.
Starting with the axiom $AB$ in $C1$, if we apply $r1.3$, then we can apply $r2.2$ only in $C2$ and that produces $\l$ in $C1$. 
The nonterminals~$A$ and~$B$ serve as markers and if an~$a$ is introduced to the left of~$A$ in $C1$ (by $r1.1$), then  one~$a$ is introduced to the left of~$B$ (by $r2.1)$ in $C2$. Similarly, if one~$b$ is introduced to the left of~$A$ in $C1$ (by $r1.2$), then a~$b$ is introduced to the left of~$B$ (by $r3.1$) in $C3$. This guarantees to have the pattern of the copy language produced by the system~$\Pi$. But, there is a caveat here. If one can apply $r1.1$ in $C1$, then in $C2$, $r2.2$ can also be applied and in such a case, the pattern of the copy language is not followed.  However, then back in $C1$, only $r1.3$ can be applied, which means for the string to move to $C2$ and there, the derivation stops. As $C2$ is not a final component, by definition the strings over the terminal alphabet $\{a,b\}$ that are not also leading into the final component are not collected into the language generated by~$\Pi$. 
\par A sample derivation for $aabaab$ is given below. 
\beq
(AB)_1 \Ra_{r1.1} (aAB)_2 \Ra_{r2.1} (aAaB)_1 \Ra_{r1.1} (aaAaB)_2 \Ra_{r2.1} (aaAaaB)_1 \Ra_{r1.2} (aabAaaB)_3 \\ \Ra_{r3.1} (aabAaabB)_1 \Ra_{r1.3} (aabAaab)_2 \Ra_{r2.2} (aabaab)_1.
\eeq
The control graph underlying the construction is shown in Figure~\ref{figCONTROLgraphs}.
\end{example} 
\begin{table}[tb]
\begin{center}
\begin{tabular}{|c|c|c|}
\hline
{\bf Component C1}&{\bf Component C2} & {\bf Component C3} \\
\hline
$r1.1: (1,(\l,a,A)_I, 2)$ & $r2.1: (2,(\l,a,B)_I,1)$& $r3.1: (3, (\l,b,B)_I,1)$ \\ 
$r1.2: (1,(\l,b,A)_I,3)$ & $r2.2: (\l,A,\l)_D,1)$ &  \\
$r1.3: (1,(\l,B,\l)_D,2)$ && \\
\hline
\end{tabular}
\caption{Star-controlled ins-del system for generating 
$\text{Copy}_{a,b}=\{ ww \mid w \in \{a,b\}^*\}$.}
\label{tableCOPYlang}
\end{center}
\end{table}

\begin{figure}[tb]
\centering
\begin{subfigure}[t]{0.4\textwidth}
\begin{tikzpicture} 
\filldraw[black, draw=black]
    (1.5,4) circle (2pt) node[above] {$C3$}
    (0,4) circle (2pt) node[above] {$C1$}
    (-1.5,4) circle (2pt) node[above] {$C2$};
\draw   
 (-1.5,4) -- (1.5,4);
\end{tikzpicture}
\caption{Star-shaped control graphs underlying the star-controlled systems depicted in Example~\ref{exa-copy}.} 
\label{fig3star}
\end{subfigure}
\quad 
\begin{subfigure}[t]{0.4\textwidth}
\begin{tikzpicture} 
\filldraw[black, draw=black]
    (5.5,4.8)circle (2pt) node[above] {$C6$}
    (8.5,4.8) circle (2pt) node[above] {$C4$}
    (8.5,4) circle (2pt) node[below] {$C3$}   
    (7,4.8) circle (2pt) node[above] {$C5$}
    (7,4) circle (2pt) node[below] {$C1$}
    (5.5,4) circle (2pt) node[below] {$C2$};
\draw   
 (5.5,4) -- (8.5,4)
 (7,4) -- (8.5,4.8);
 \draw[dashed] (5.5,4.8) -- (7,4)
-- (7,4.8);
 
\end{tikzpicture}\caption{Control graphs underlying the star-controlled systems  in our main theorems.} 
\label{fig4star}
\end{subfigure}
\caption{Control graphs underlying different GCID systems in this paper.}
\label{figCONTROLgraphs}
\end{figure}
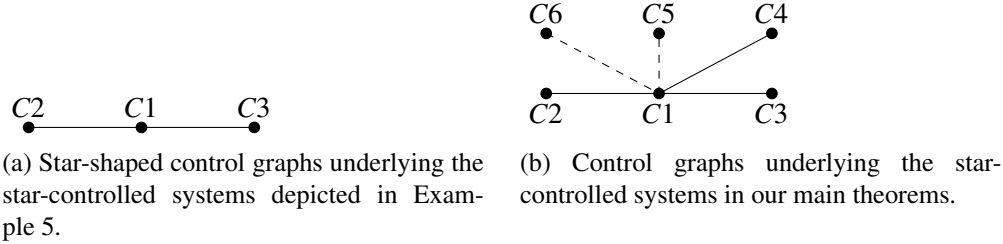

\LV{Incidentally, when the underlying control graph forms a star structure, this system can be interpreted as an insertion-deletion P system \cite{KriRam0102}, where the outer membrane $C1$ contains several disjoint inner membranes $C_2, C_3, \dots C_k$. Certain  objects (in our case, the simulating rules) placed in the outer membrane. When a rule $(1,r,j)$  available in the outer membrane is applied, the control is sent then sent to the $j^{th}$ membrane and any rule in $C_j$ can then be applied. For more details, see \cite{Pau2002}, \cite{AlhFreIva2014} on length P systems, where a linear membrane structure  with multiset objects in the form of vectors/numbers is considered.}

\section{Computational Completeness}\label{sec-REresults}
In this section, we present the main results of our paper.
First, we discuss some limitations for getting computational completeness results and then, we describe two important cases of resource restrictions that characterize $\mathrm{RE}$.

\subsection{GCID$_S$ systems with insertion and deletion length one}

In \cite{TakYok2003}, it has been proved that ins-del systems with size (1,1,1;1,1,1) characterize $\mathrm{RE}$. 
\LV{If we desire to have one-sided context for insertion/deletion,  then}\SV{Notice that}
it is proved in \cite{KraRogVer2008,MatRogVer2007} that  ins-del systems of size $(1,1,1;1,1,0)$ or $(1,1,0;1,1,1)$ cannot characterize $\mathrm{RE}$. It is therefore obvious that we need at least $2$ components in a graph-controlled ins-del system of sizes $(1,1,1;1,1,0)$ and $(1,1,0;1,1,1)$ to characterize $\mathrm{RE}$. In \cite{FerKupRam2017b},  we characterized $\mathrm{RE}$ by path-controlled GCID systems of size $(k;1,i',i'';1,j',j'')$ for different combinations of $k \geq 1, i', i'', j', j''$.  

However, if we impose star structure as the underlying control graph and the resultant string has to move in/move out during every derivation step, then it is interesting to notice that the context-free rules of SGNF, namely $p: X \to bY$, $q: X \to Yb$ and $h: S' \to \l$ can never be \emph{directly simulated} by rules of $\mathrm{GCID}_S(k;1,i',i'';1,j',j'')$ for any value of $k \geq 2, i',i'',j',j'' \geq 0$. 
Here, by a \emph{direct simulation} 
of a rule~$r$, we mean that, assuming a sentential form~$w$ may yield the sentential form~$v$ by applying rule~$r$ within the original grammar (which is, in our case, in SGNF), then the simulating star-controlled GCID system will start in the central component $C1$ with the sentential form~$w$ and derive after a number of steps, possibly, during the simulation, introducing and deleting symbols specific to~$r$ (called \emph{markers} in the following), the sentential form~$v$ and moving back to~$C1$ to be ready for the next simulation step.
\begin{proposition}
\label{star-not-CF}
The context-free rules of a grammar in SGNF, namely $p: X \to bY$ and $q: X \to Yb$ (with $X\neq Y$), as well as  $h: S' \to \l$, can never be directly simulated by rules of $\mathrm{GCID}_S(k;1,i',i'';1,j',j'')$ for any value of $k \geq 2, i',i'',j',j'' \geq 0$.     
\end{proposition}


\newcommand{\proofofstarnotCF}{
\noindent\begin{proof}
To directly simulate $p: X \to bY$ using insertion-deletion rules, we need  two insertion rules (one to insert $b$ and one to insert $Y$; here we recall that the insertion length is $1$) and  one deletion rule to delete~$X$. Hence, we need three \emph{basic} insertion-deletion rules. Further, if we need to introduce $r\geq 1$ markers, then we can insert only one marker at a time using an insertion rule which will account for $r$ insertion rules. At the end of the derivation, we need to delete all the $r$ markers using $r$ deletion rules (since we can only delete one symbol at a time). This amounts to having $r$ insertion rules and $r$ deletion rules to deal with the markers and $3$ basic insertion-deletion rules to simulate $p: X \to bY$. This sums up to $2r+3$ rules. 

We need to distribute these $2r+3$ rules among the $k$ components of the GCID system. Let $C1$ be the central (initial and final) component. 
As the system is star-structured, in a rule $(i,(x,y,z)_\delta,j), \delta\in\{I,D\}, i,j \in [1\ldots k]$, we have $i\neq j$, as loops are not allowed, and $|\{i,j\}\cap\{1\}|=1$. Hence, the order of rule applications in any derivation will start at the central node and then alternate between central and non-central nodes.
Therefore, (i) the last rule of the simulation should be placed in a non-central component and not in $C1$ and (ii) the total number of rules for simulation is even. Since $2r+3$ is not even, the statement follows. \end{proof}}

\SV{\proofofstarnotCF}

By its definition\LV{ and as mentioned in the previous proof}, a derivation of a GCID$_S$ system has to alternate between the central component and any other component. By putting exactly the same rules in two components and declaring one of the two components as being central, while both are final, one obtains:
\begin{proposition}
$\mathrm{GCID}(1;n,i',i'';m,j',j'')\subseteq \mathrm{GCID}_S(2;n,i',i'';m,j',j'')$ holds for any value of $i',i'',j',j'' \geq 0$ and $n,m\geq 1$.   
\end{proposition}

For example, this entails $\mathrm{GCID}_S(2;1,1,1;1,1,1)=\mathrm{RE}$ and similar computational completeness results based on what is known for ins-del systems. By way of contrast, computational completeness results for 2-component graph-controlled systems do not necessarily carry over to our star-controlled systems, as there, `loops' might be allowed. 


\subsection{GCID$_s$ systems with insertion or deletion length of more than one}

To simplify the presentation and proofs of our further results, the following observations from~\cite{FerKupRam2017b} are used, adapted to our case.

\begin{proposition} \label{cor-reversal} \cite{FerKupRam2017b}
Let 
$k,n,i',i'',m,j,j''$ be non-negative integers.
\LV{The following statements are true.} 
\begin{enumerate}
\item $\mathrm{GCID}_S(k;n,i',i'';m,j',j'')
= [\mathrm{GCID}_S(k;n,i'',i';m,j'',j')]^R$;
\item $\mathrm{RE}=\mathrm{GCID}_S(k;n,i',i'';m,j',j'')$
iff 
$\mathrm{RE}=\mathrm{GCID}_S(k;n,i'',i';m,j'',j')$.
\end{enumerate}
\end{proposition}

\begin{table}[tb]
\begin{center}
\scalebox{0.9}{
\resizebox{1\linewidth}{!}{
\begin{tabular}{|c|c|c|}
\hline
{\bf  $C1$}&{\bf $C2$}&{\bf $C3$} \\\hline
$p1.1: (1,(X,p,\l)_I,2)$ & $p2.1: (2,(\l,X,\l)_D,1)$ & $p3.1: (3,(p',p^{v},\l)_I,1)$\\ 
$p1.2: (1,(p,p',\l)_I,3)$ & 
$p2.2: (2,(\l,p''p''',\l)_D,1) $ & \\
$p1.3: (1,(p',b,\l)_I,5)$ & &\\
$p1.4: (1,(\l,pp',\l)_D,4)$ &&\\
$p1.5: (1,(p''',p^{iv},\l)_I,4)$ &&\\
$p1.6: (1,(p''',Y,\l)_I,2)$ & &\\
\hline
$q1.1: (1,(X,q,\l)_I,2)$ & $q2.1: (2,(\l,X,\l)_D,1)$ &$q3.1: (3,(q',b,\l)_I,1)$\\
$q1.2: (1,(q,q',\l)_I,3)$ & & \\
$q1.3: (1,(q',Y,\l)_I,4)$ & & \\
\hline
$h1.1: (1,(S',h,\l)_I,2)$ &$h2.1: (2,(\l,S'h,\l)_D,1)$ &\\
\hline
$f1.1: (1,(A,f,\l)_I,6)$ &$f2.1: (2,(\l,Af',\l)_D,1)$& \\
$f1.2: (1,(\l,fB,\l)_D,2)$ &&\\
\hline
$g1.1: (1,(C,g,\l)_I,6)$ &$g2.1: (2,(\l,Cg',\l)_D,1)$& \\
$g1.2: (1,(\l,gD,\l)_D,2)$ &&\\
\hline
\hline
{\bf  $C4$}&{\bf  $C5$}&{\bf  $C6$}\\
\hline 
 $p4.1: (4,(p'',p''',\l)_I,1)$ & $p5.1: (5,(b,p'',\l)_I,1)$ & \\
$p4.2: (4,(\l,p^{iv}p^{v},\l)_D,1)$
&& \\
\hline
$q4.1: (4,(\l,qq',\l)_D,1)$ && \\
\hline && \\
\hline
&&$f6.1: (6,(B,f',\l)_I,1)$ \\ \hline
&&$g6.1: (6,(D,g',\l)_I,1)$ \\ \hline

\end{tabular}}}
\caption{Star-controlled $\mathrm{GCID}_S(6;1,1,0;2,0,0)$ simulating the rules of SGNF.
}
\label{table-RE-6110200}
\end{center}
\end{table}

\begin{theorem} \label{RE6110200} $\mathrm{RE} = \mathrm{GCID}_S(6;1,1,0;2,0,0)$ and $\mathrm{RE} =\mathrm{GCID}_S(6;1,0,1;2,0,0)$.  \end{theorem}
\begin{proof} Consider a type-$0$ grammar $G=(N,T,P,S)$ in SGNF as in Definition~\ref{def-SGNF}. 
\LV{The rules of $P$ are\LV{ assumed to be} labelled bijectively with labels from $[1 \ldots |P|]$. }We construct a\LV{ star-controlled} GCID system $\Pi= (6,V,T,\{S\},H,1,\{1\},R)$ of size $(6;1,1,0;2,0,0)$\LV{, with a control graph shown in Fig.~\ref{fig4star}, as follows} such that $L(\Pi) = L(G)$. The alphabet $V$  contains rule markers, apart from the symbols of~$G$. More specifically, for each rule $\gamma\to\delta\in P$ labeled~$r$, we have $r\in V$. Moreover, if $\gamma\to\delta\neq S'\to\l$, we also have single-primed markers $r'\in V$. Finally, for context-free rules of the form $X\to bY$, even markers $r'', r''', r^{iv},r^v\in V$.
We refer to Table~\ref{table-RE-6110200}, showing the simulation of the different rule types  of SGNF. 
The columns of the table correspond to the components of~$\Pi$. The rows of Table~\ref{table-RE-6110200} correspond to the rules simulating the `linear rules' $p: X \to bY$ and $q: X \to Yb$, with $X\in N'$ and $b\in N''\cup T$, $h: S' \to \l$, as well as $f: AB \to \l$ and $g: CD \to \l$. 
We now prove that $L(G) \subseteq L(\Pi)$ as follows. We show that
if $w\Ra w'$ in~$G$, then $(w)_1\Ra_* (w')_1$ according to $\Pi$. The claim then follows by induction.


\smallskip
\noindent\underline{Context-free rule $q:X\to Yb$.} Here, $w=\a X\b$ and $w'=\a Yb\b$ for some $\a\in \{A,C\}^*$, $\b\in (\{B,D\}\cup T)^*$. The simulation performs as follows:
\begin{eqnarray*}(\a X\b)_1 &\Ra_{q1.1}& (\a Xq\b)_2 \Ra_{q2.1} (\a q\b)_1 \Ra_{q1.2}
 (\a qq'\b)_3 \Ra_{q3.1} (\a qq'b\b)_1 \\& \Ra_{q1.3} &
 (\a qq'Yb\b)_4 \Ra_{q4.1}  (\a Yb\b)_1 \,.\end{eqnarray*}
\smallskip
\noindent\underline{Context-free rule $p:X\to bY$.} Here, $w=\a X\b$ and $w'=\a bY\b$ for some $\a\in \{A,C\}^*$, $\b\in (\{B,D\}\cup T)^*$. The simulation performs as follows:
\begin{eqnarray*}(\a X\b)_1 &\Ra_{p1.1}& (\a Xp\b)_2 \Ra_{p2.1} (\a p\b)_1 \Ra_{p1.2}
 (\a pp'\b)_3 \Ra_{p3.1} (\a pp'p^v\b)_1  \\
 & \Ra_{p1.3}&(\a pp'bp^v\b)_5 \Ra_{p5.1} (\a pp'bp''p^v\b)_1 \Ra_{p1.4}
 (\a bp''p^v\b)_4 \Ra_{p4.1} (\a bp''p'''p^v\b)_1 \\  &\Ra_{p1.5}&(\a bp''p'''p^{iv}p^v\b)_4 \Ra_{p4.2}  (\a bp''p'''\b)_1 \Ra_{p1.6} (\a bp''p'''Y\b)_2 \Ra_{p2.2} (\a bY\b)_1 \,.\end{eqnarray*}


One might wonder that the $h: S'\to\l$ and the $f: AB \to \l$  rules could easily be simulated by the rules $(1,(\l,S',\l)_D,1)$ and $(1,(\l,AB,\l)_D,1)$, respectively. However, the underlying control graph of our star-controlled GCID forbids loops and hence, we have given a different simulation for these rules. Since the correctness of the $h$-rule simulation is trivial, it remains to discuss the simulation of the  rule $f: AB \to \l$. 
The simulation of $g: CD\to \l$ is similar and hence omitted.

\smallskip
\noindent\underline{Non-context-free rules $f:AB\ra\l$.}
This means that we expect $w=\a AB\b$ and $w'=\a\b$ for some $\a\in \{A,C\}^*$, $\b\in (\{B,D\}\cup T)^*$. 
Within $\Pi$, this can be simulated as follows.
\begin{align*}(\a AB\b)_1 \Ra_{f1.1} (\a AfB\b)_6 \Ra_{f6.1}
(\a AfBf'\b)_1 \Ra_{f1.2}
 (\a Af'\b)_2 \Ra_{f2.1} (\a \b)_1\,.
 \end{align*}

\smallskip
Conversely, a derivation $(w)_1\Ra^*(w')_1$, with $w\neq w'$ and $w,w'\in \{A,C\}^*\linebreak[3](N'\cup\{\emptyword\})(\{B,D\}\cup T)^*$  has to start like $(w)_1\Ra (v)_j$ in~$\Pi$. If some rule from $C1$ is applied to $w$, the rule will insert a rule marker into the string~$w$ and branch to $C2$ (when simulating context-free rules) or to $C3$ (when simulating non-context-free rules). The introduction of rule markers in $C1$ will take care of the non-interference among the non-context-free and context-free rules. We now discuss the possibilities in detail. In our discussion, we distinguish between sentential forms containing or not containing a symbol from $N'$. 

Our inductive arguments will also show that, in~$\Pi$, no sentential form is derivable that contains two occurrences of symbols from~$N'$. More generally, we can show the following. 
Assume that we can derive some configuration $(w)_1$ in~$\Pi$ such that the string~$v$ contains no marker symbols, where $v=\mu(w)$ is obtained by applying the morphism~$\mu$ that acts like the identity on~$V$ apart from the letters $f,f',g,g'$ that will be erased. Then the sentential form~$v$ is also derivable in~$G$. In particular, if~$w\in T^*$, then $v=w$, i.e., each word in $L(\Pi)$ also belongs to $L(G)$.

We will also prove by induction that, if $(w)_1$ is derivable in~$\Pi$ and if $w$ contains at most one occurrence of~$N'$ and no markers but possibly $f,f',g,g'$, then also $\mu(w)$ is derivable in~$\Pi$ and then, only using the rules $f1.1$, $f6.1$, $g1.1$ and $g6.1$, we can derive~$\mu(w)$ in~$\Pi$ from~$w$. This also means that, in each such string~$w$ derivable in~$\Pi$, the number of occurrences of symbols from $\{f,g\}$ equals the number of occurrences of symbols from $\{f',g'\}$. We also call this the \emph{balance condition}. Therefore, we can start our inductive hypothesis with strings that can be derived both in~$G$
 and in~$\Pi$ and observe the maintenance of the balance condition along our arguments.

Let us first assume (by induction) that the sentential form $w^1 = \a X\b$ for some $\a\in \{A,C\}^*$, $\b\in (\{B,D\}\cup T)^*$ and $X \in N'$ is derivable in~$G$ and the configuration $(w^1)_1$ is derivable in~$\Pi$. We will prove (as induction step) that if $(w^1)_1\Ra_*(u)_1$, $w^1\neq u$, and $u\in \{A,C\}^*(N'\cup\{\lambda\})(\{B,D\}\cup T)^*$ is the first sentential form from $\{A,C\}^*(N'\cup\{\lambda\})(\{B,D\}\cup T)^*$ that appears in a derivation of~$\Pi$ in $C1$, then $w^1\Ra u$ holds in~$G$, except from a premature start of simulating non-context-free rules\LV{ that we} also discuss below and where we argue that the balance condition is maintained.


\smallskip\noindent
\underline{Applying $f1.1$} to $w^1 = \a X\b$ for some $\a\in \{A,C\}^*$, $\b\in (\{B,D\}\cup T)^*$ and $X \in N'$, the only nonterminal from $N'$, is possible but unintended. (A similar discussion applies to $g1.1$.)
A successful application yields the configuration $(w^2)_6$, with $w^2=\alpha_1 Af\alpha_2 X\b$, where
$\alpha_1 A\alpha_2=\alpha$.
Now, the only applicable rules are $f6.1$ or $g6.1$.
We get a configuration $(w^3)_1$ with $\mu(w^3)=w^1$. It can be observed that on $w^3$, neither $f1.2$ nor $g1.2$ are applicable, as these rules require $AB$ or $CD$ sitting in the center of $w^1$, which was not the case by assumption. Therefore, we could either continue our journey with inserting further markers $f,f',g,g'$ (but always maintaining the balance condition) or finally apply $r1.1$, belonging to a context-free rule~$r$. Now, observe that, instead of applying $r1.1$ and then $r2.1$ (because $f2.1$ and $g2.1$ are inapplicable), yielding a configuration $(u)_1$, we could also first apply $r1.1$ and $r2.1$ to $w^1$, and then the same $f$- and $g$-simulation rules as before, arriving at $(u)_1$ in a different way. This proves (here as part of  the induction step) that, as claimed, we  can exchange the sequence of rule applications in a way that we apply $f$- and $g$-simulation rules after the other rules that are meant to simulate the context-free rules. We also see by induction that the balance condition is always maintained.


\smallskip\noindent
\underline{Applying $q1.1$} to $w^1 = \a X\b$ for some $\a\in \{A,C\}^*$, $\b\in (\{B,D\}\cup T)^*$ and $X \in N'$, the only nonterminal from $N'$, yields the configuration $(w^2)_2$ with $w^2=\a Xq\b$. In $C2$, all rules (except $q2.1$ and $p2.1$) are guarded by markers and the only applicable rule are $q2.1$ or $p2.1$ (in case there is a rule $X \to b'Y'$ in~$P$) which delete~$X$, yielding $(w^3)_1=(\a q\b)_1$.  Due to the rule markers, apart from the intended rule $q1.2$, one could also apply $f1.1$ or $g1.1$. In that case, having moved to $C6$, $f6.1$ or $g6.1$ are applicable, yielding some configuration $(w')_1$ with $\mu(w')=w^3$. As discussed before, we can even continue like this, but in order to make any further progress, we will have to apply $q1.2$ in some configuration $(w'')_1$ with $\mu(w'')=w^3$. As we could also apply the non-context-free simulation rules afterwards, it suffices to discuss what happens if we apply $q1.2$ to $w^3$ as intended. Hence, we arrive at the configuration $(w^4)_3$ with $w^4=\a qq'\b$. 
The required rule markers cause 
$q3.1$ to be the only applicable rule as desired. 
Therefore, we arrive at the configuration $(w^5)_1=(\a qq'b\b)_1$. Clearly, one could now (again) apply $f1.1$ or $g1.1$, but this would only lead to prematurely introducing the markers $f,f',g,g'$ similar as discussed before, again always maintaining the balance condition. Therefore, the only applicable rule that needs to be discussed (apart from the intended one, which is $q1.3$) is $q1.2$ (again). With the string $\a qq'q'b\b$, we are back to $C3$. Now, there are two possible subsequent configurations: (a) $(\a qq'bq'b\b)_1$, or (b) $(\a qq'q'bb\b)_1$. In Case (a), we claim that there is no way to delete the second occurrence of $q'$ in the future. Namely, the only way to delete $q'$ is if left to it, $q$ is sitting. But as now some $b\in\{B,D\}\cup T$ is to the left of $q'$, there is no way to introduce $q$ in this position later, because the marker~$q$ always works as a symbol that replaces the former $N'$-symbol. Therefore,  a derivation following (a) cannot terminate. The situation is different in Case (b). For instance, we can apply $q1.3$ to string $\a qq'q'bb\b$, followed by $q4.1$.
Again, we have two configurations to study: (A)  $(\a q'Ybb\b)_1$ or (B)  $(\a Yq'bb\b)_1$. In Case (A), we can argue similarly to Case (a) above to see that this configuration cannot lead to a terminal string: left to~$q'$ will sit some symbol $A$ or~$C$. Case (B) is indeed different. Assuming that only rules of the form $Z\to b'Z'$ are simulated subsequently, there may be a derivation $Y\Ra^*\gamma X$ with $\gamma\in\{A,C\}^+$ that is simulated by the GCID system as intended. Hence, we see now a configuration $(\a \gamma Xq'bb\b)_1$ and then, after a short excursion into $C2$, we see $(\a \gamma qq'bb\b)_1$. Now, we can actually terminate by using the rules $q1.3$ and $q4.1$, leading to $(\a \gamma Ybb\b)_1$.
However, we would arrive at the same string if we had followed our intended plan. Then, we could get from  $(\a Yb\b)_1$ via $(\a \gamma Xb\b)_1$ to $(\a \gamma qb\b)_1$. Now, after applying $q1.2$ and $q2.2$ as intended, we can also see $(\a \gamma qq'bb\b)_1$ and continue as above. This argument is also valid (by a separate yet straightforward induction) if we happen to produce a string $(\a q(q')^kb^k\b)_1$ for an arbitrary $k>1$. Therefore, we can avoid this process that we call \emph{rule inversion}, and always follow our standard derivation instead.
We can hence assume that we apply $q1.3$ to $w^5$ as desired. 
Therefore, we arrive at the configuration $(w^5)_1=(\a qq'Yb\b)_4$. If we actually apply $q4.1$, then we arrive at $(\a Yb\b)_1$ as intended, proving the inductive step in this case. 
But by the very structure of this component, no other rules are applicable.

\smallskip\noindent
\underline{Applying $p1.1$} to $w^1 = \a X\b$ for some $\a\in \{A,C\}^*$, $\b\in (\{B,D\}\cup T)^*$ and $X \in N'$, the only nonterminal from $N'$,  will insert a marker $p$ to the right of $X$, yielding $(w^2)_2=(\a Xp\b)_2$. Recall that we are trying to simulate the rule $X\to bY$ for some $X,Y\in N'$ and $b\in\{A,C\}$.
In $C2$, all rules (except $p2.1$ and $q2.1$) are guarded by markers and the only applicable rule are $q2.1$ or $p2.1$ (in case there is a rule $X \to b'Y'$ in $P$) which delete the nonterminal~$X$, yielding a string $w^3=\a p\b$, i.e., we arrive at the configuration $(w^3)_1$. \LV{The application of $p1.1$ and $p2.1$ in sequence corresponds to the direct simulation of the rewriting rule $X \to p$. The string is now in $C1$. }Since $X$ is deleted in the previous step and there is no $p', p''$, the only applicable rule is $p1.2$ which inserts a $p'$ after $p$, yielding the configuration $(w^4)_3=(\a pp'\b)_3$\LV{, moving $w^3$ into $C3$}. In $C3$, guarded by rule markers, we have to apply $p3.1$ as intended. 
Hence, we arrive at $(w^5)_1=(\a pp'p^v\b)_1$.
If we re-apply $p1.2$, we achieve an imbalance of the number of occurrences of~$p$ and~$p'$. This is problematic insofar, as $pp'$ is deleted together. Also, we would have to then re-apply $p3.1$ again
, creating another imbalance.
Alternatively, we could apply $f1.1$ or $g1.1$, which 
introduces a  pair of non-primed and primed $f$/$g$-markers prematurely, but maintaining their balance
. \LV{All these considerations }\SV{This }bring\SV{s} us to the conclusion that we should\LV{ rather} apply $p1.3$ in $(w^5)_1$. 

Hence, we arrive at  $(w^6)_5=(\a pp'bp^v\b)_5$, with $b\in \{A,C\}$. In $C5$, we have to apply a rule that puts some marker $r''$ to the right of an occurrence of~$b$. As $b\in \{A,C\}$, the $b$ occurring between $p'$ and $p^v$ is the rightmost of all occurrences of $A,C$ within $w^6$. 
Let us first discuss what happens if we do apply some $r5.1$ (but possibly $r\neq p$) to this described rightmost occurrence and mark the situation when some~$b$ within~$\a$ is affected as $(*)$, not to forget its discussion. We get to $(w^7)_1=(\a pp'br''p^v\b)_1$. Apart from now applying $f1.1$ or $g1.1$, we could also apply $p1.2$ or $p1.3$, and finally also $p1.4$ (as intended). The scenario of prematurely introducing the $f$/$g$-markers has been sufficiently discussed before. If we apply rule~$p1.2$, we again create an imbalance concerning $p/p'$. Let us defer the discussion of applying $p1.3$ at this configuration  $(w^7)_1$ a bit as $(+)$; we rather discuss applying $p1.4$. We arrive at the configuration 
$(w^8)_4=(\a br''p^v\b)_4$. Now, only $r4.1$ would be applicable, leading to $(w^9)_1=(\a br''r'''p^v\b)_1$.
We are now in safer waters, as 
we have to use rule $r1.5$ to get to $(w^{10})_4=(\a br''r'''r^{iv}p^v\b)_4$, because if we apply $p1.6$ directly, we have no chance to delete $p^v$ in the future. If we apply $r4.1$ again on $w^{10}$, we create an imbalance between the number of $r''$ and $r'''$, but this balance is necessary for deleting these markers in~$C2$.
By using $r4.2$ alternatively on $w^{10}$, one can see that the only chance to continue the route is when we have $r=p$.
In that case, we move to $C1$ with $w^{11}=\a bp''p'''\b$. If we now re-apply $p1.5$,  at $C4$, we have to apply $p4.1$ and create an imbalance between $p''$ and $p'''$, hence preventing us from a terminating derivation. If we introduce $f$- or $g$-markers, we are  forced  to introduce primed versions in $C6$; we have discussed these premature but balanced introductions of these markers before. 
Hence, we have to discuss applying $p1.6$ as intended.
We enter $C2$ with $w^{12}=\a bp''p'''Y\b$.
Now, we can either delete $Y$ with some fitting rule $s2.1$ and return to the configuration $(w^{11})_1$, hence making no progress, or we apply $p2.2$ as intended, finally getting to the configuration $(w^{12})_1=(\a bY\b)_1$ as desired. 

In order to conclude that the induction step has been shown, we still have to consider two scenarios, marked as $(*)$ and $(+)$ above. 
In\LV{ scenario} $(+)$, we look at \SV{$}\LV{$$}(w^7)_1=(\a pp'br''p^v\b)_1\Ra_{p1.3}(\a pp'bbr''p^v\b)_5\,.\LV{$$}\SV{$}
Assume we apply a rule $s5.1$ next. As  the case when we find $b\in\{A,C\}$ within $\a$ is similar to the discussion $(*)$ that is still to come, we focus on two cases of configurations: (1) $(\a pp'bs''br''p^v\b)_1$ or (2) $(\a pp'bbs''r''p^v\b)_1$.
In both configurations, we can again apply $p1.3$, but this makes the whole case fail even more. 
We can now derive (under the conclusion that $r=p$) in the same way as in the main line of derivation, leading to $(\a bs''bY\b)_1$ (Case (1)) or to 
$(\a bbs''Y\b)_1$ (Case (2)). In both cases, there is no way to make use of $s''$, because this means we have to move to $C4$, or we mis-use another $p$-type rule at some point, when $p1.4$ makes us enter $C4$ again, but then continuing with the $s$-markers (using $s4.1$). Let us clarify this by assuming that we simulate $t:Y\to b'Y'$ next. Following the standard simulation up to $t1.4$, we get  $(\a bs''bb't''t^v\b)_4$ (Case (1)) or 
$(\a bbs''b't''t^v\b)_4$ (Case (2)). We could now use $s4.1$, $s1.6$ and $s2.2$ to introduce another nonterminal from $N'$ at the position of $s''$, but behold: we have now another left-over double-primed marker $t''$ whose removal can only be achieved by switching between two rule simulations in the `next round'. Therefore, we will never be able to terminate this derivation.

For scenario $(*)$, we reconsider  $(w^6)_5=(\a pp'bp^v\b)_5$, with $\a=\a_1 b\a_2$, so that for a suitable rule~$r$ that should introduce $Z\in N'$,  $(w^7)_1=(\a_1 br''\a_2 pp'bp^v\b)_1$. We could try to continue with $(w^7)_1\Ra_{p1.4}$
$$ (\a_1 br''\a_2 bp^v\b)_4\Ra_{r4.1}(\a_1 br''r'''\a_2 bp^v\b)_1\Ra_{r1.6} (\a_1 br''r'''Z\a_2 bp^v\b)_2\Ra_{r2.2} (\a_1 bZ\a_2 bp^v\b)_1$$
but then there is never a chance to lose $p^v$ again
. Therefore, also this scenario will never see a  derivation producing a terminal string.

\smallskip\noindent
We are now discussing a string $w$ derivable in $G$ and as configuration $(w)_1$ in~$\Pi$, with $w=\a AB\b$, with $\a\in \{A,C\}^*$ and $\b\in(\{B,D\}\cup T)^*$. The case of a string of the form $\a CD\b$ can be discussed in a very similar fashion.
First observe that we cannot apply any rule $p1.x$ or $q1.x$ or $h1.x$ due to the absence of nonterminals from $N'$ or of required markers.
We could in fact start with $g1.1$, followed by $g6.1$, and even repeat this, so that some $g$-markers are attached to $C$-occurrences. Similarly\LV{ as argued above}, we can consider such derivations to occur prematurely, because finally we have to use the $f$-rule as explained next.

\smallskip\noindent
\underline{Applying $f1.1$} to $w=\a AB\b$, we get a string $w_1$ by inserting~$f$ anywhere after an $A$-occurrence within~$w$. Let $\a A=\a_1A\a_2$ indicate this position, i.e., $w_1=\a_1 Af\a_2B\b$.
$w_1$ is transferred to component~$C6$. So, $f6.1$ is applied and the 
string, yielding $w_2=\alpha_1 Af\a_2\b_1 Bf'\beta_2$, which enters $C1$, where $B\b=\b_1 B\beta_2$. Notice that the configuration $(w_2)_1$ could have also been created by  a premature application of $f1.1$ and $f6.1$ in some earlier phase of the derivation. This explains how a string that satisfies the balance condition could finally yield a terminal string, although it is not following the standard simulation as described in the beginning of the proof.
As $\a_2$ cannot contain any $B$-occurrence, now applying $f1.2$ necessitates $\a_2=\b_1=\emptyword$, and then, $w_3=\a_1 Af'\b_2$ is sent to $C2$.
There, the only applicable rule is $f2.1$ as intended, producing $w_4=\a_1\b_2$, sent to $C1$ as intended. As mentioned at several places, instead of applying $f1.2$ on $w_2$, one could also possibly apply $f1.1$ again, or also $g1.1$. We can consider all these attempts as premature ones, they only affect the left part of the string and have to be finally successfully matched by using rules $f1.2$ or $g1.2$, followed by executing another deletion rule in $C2$.

\smallskip
\noindent
It could be that a string $w=\a\b$ was derived in~$G$ (and hence possibly the configuration $(w)_1$ in~$\Pi$ by induction) with $\a\in\{A,C\}^*$ and $\b\in(\{B,D\}\cup T)^*$ and neither $\a$ ends with $A$ and $\b$ starts with~$B$ nor $\a$ ends with~$C$ and $\b$ starts with~$D$. We can still apply rules $f1.1$ or $g1.1$, moving the resultant string to $C6$, where $f6.1$ or $g6.1$ are applicable, moving us back to $C1$. Yet, the crucial observation is that neither $f1.2$ nor $g1.2$ are ever applicable now, as they require the presence of the substring $AB$ or $CD$ (within~$w$), respectively. Only then, the substrings $fB$ or $gD$ can be created.

\smallskip
This concludes our argument concerning the inductive step of the correctness proof of our suggested simulation.
%

Finally, Proposition~\ref{cor-reversal} shows that  star-controlled GCID systems of size $(6;1,0,1;\linebreak[2]2,0,0)$ are computationally complete, as well.  
\end{proof}
\newcommand{\remarkonprule}{\begin{remark}\label{rem:p-notsimple}
The reader might wonder if it would be possible to merge some of the components of the previous construction (Theorem \ref{RE6110200}), but this will create malicious derivations in each case. 
Also, the simulation of a $p$-rule cannot follow \LV{the same}\SV{a} simple pattern as that of the $q$-rule (see Rem.~\ref{rem:p-notsimple}), as we want to avoid the derivation of strings with more than one occurrence of a symbol from~$N'$. 
Here, we explain why a simple, not complex looking and seeming correct $p: X \to bY$ rule simulation does not work with the size $(5;1,1,0;2,0,0)$. Consider, if we attempt to construct a $p$-rule simulation for $\Pi$ as in Table~\ref{table-attempt-prule-5C}. 

\begin{table}[tb]
\begin{center}
\resizebox{.95\linewidth}{!}{
\begin{tabular}{|c|c|c|c|c|}
\hline
{\bf Component C1}&{\bf Component C2}&{\bf Component C3}&{\bf Component C4}&{\bf Component C5} \\
\hline
$p1.1: (1,(X,p,\l)_I,2)$ & $p2.1: (2,(\l,X,\l)_D,1)$ & $p3.1: (3,(\l,p'',\l)_D,1)$& $p4.1: (4,(p,p',\l)_I,1)$ & $p5.1:(5,(\l,pp',\l)_D,1)$\\ 
$p1.2: (1,(p,p'',\l)_I,4)$ & & & & \\
$p1.3: (1,(p',b,\l)_I,5)$ & && &\\
$p1.4: (1,(p'',Y,\l)_I,3)$ &&& &\\
\hline
\end{tabular}}
\end{center}
\caption{A direct simulation attempt for a $p$-rule $p:X\to bY$.}
\label{table-attempt-prule-5C}
\end{table}
A sample derivation of $p$-rule with the size $(5;1,1,0;2,0,0)$ is as follows. 
\beq (\a X \b)_1 \Ra_{p1.1} (\a Xp \b)_2 \Ra_{p2.1} (\a p \b)_1 \Ra_{p1.2} (\a pp''\b)_4 \Ra_{p4.1} (\a pp'p''\b)_1 \Ra_{p1.3}  (\a pp'bp''\b)_5 \\ \Ra_{p5.1}  (\a bp''\b)_1 \Ra_{p4} (\a bp''Y\b)_3 \Ra_{p1.4}  (\a b Y \b)_1.   
\eeq
However, the simulation does not always work in the intended way as one need not apply p1.3 and instead $p1.4$ can be applied first. Therefore, the corresponding $b$ is not inserted, however the $Y$ has been inserted. With suitable a $Y$-rule that (finally) creates~$X$ again, later one could eliminate the markers $pp'$ together and that will be a problem as a malicious string could be generated. For example, consider the grammar $G$ contains the rules $p: X \to bY$ and $u: Y \to b'X$, $b,b'\in\{A,C\}$, $b\neq b'$, besides some other rules. Then, with the rules of Table \ref{table-attempt-prule-5C}, we can have the following derivation.  
\beq (\a X \b)_1 \Ra_{p1.1,p2.1,p1.2,p4.1,p1.4,p3.1} (\a pp'Y\b)_1 \Ra^{*\, \text{simulating}}_{Y\to b'X} (\a pp' b'X \b)_1 \Ra_{\text{as earlier for}~X}^* (\a pp'b'pp'Y\b)_1 \\
\Ra_{p1.3}~ (\a pp'b'pp'bY\b)_5 \Ra_{p5.1} (\a b'pp'bY\b)_1 \Ra_{p1.3} (\a b'pp'bbY\b)_5 \Ra_{p5.1} (\a b'bbY\b)_1.   
\eeq
We are supposed to get $\a bb'bY \b$ with $G$, but we could derive $\a b'bbY \b$ with $\Pi$, but not in~$G$.
\end{remark}
}

\remarkonprule
Our next computational completeness result even further reduces the deletion complexity, making it context-free\LV{ at the expense of more context needed for the insertion rules}.

\begin{theorem}
\label{RE4211100}
$\mathrm{RE} = \mathrm{GCID_S(4 ;2,1,1;1,0,0)}$.
\end{theorem} 

\SV{\begin{table}[h!]
\begin{center}
\resizebox{1.0\linewidth}{!}{
\begin{tabular}{|c|c|c|c|}
\hline
{\bf Component C1}&{\bf Component C2}&{\bf Component C3} & {\bf Component C4} \\
\hline
$p1.1: (1,(\l,p,X)_I,2)$ & $p2.1: (2,(\l,X,\l)_D,1)$ & $p3.1: (2,(\l,p,\l)_D,1)$ &\\ 
$p1.2: (1,(p,bY,\l)_I,3)$  & & &\\
\hline
$q1.1: (1,(\l,q,X)_I,2)$ & $q2.1: (2,(\l,X,\l)_D,1)$ & $q3.1: (2,(\l,q,\l)_D,1)$ &\\ 
$q1.2: (1,(q,Yb,\l)_I,3)$  & & &\\
\hline
$h1.1: (1,(\l,hh',S')_I,2)$ & $h2.1: (3,(\l,S',\l)_D,1)$&  $h3.1: (3,(\l,h,\l)_D,1)$ &\\ 
$h1.2: (1,(\l,h',\l)_D,3)$  & & &\\
\hline
$f1.1: (1,(\l,f',A)_I,2)$ & $f2.1: (2,(A,f,B)_I,1)$ & $f3.1: (3,(3,(\l,f'',\l)_D,1)$ & $f4.1: (4,(f',f''f^{2},f)_I,1)$\\ 
$f1.2: (1,(\l,A,\l)_D,4)$ & $f2.2: (2,(B,f^4,\l)_I,1)$&  
$f3.2:  (3,(\l,f',\l)_D,1)$ &$f4.2: (4,(f^2,f'''f^3,f^4)_I,1) $\\
$f1.3: (1,(\l,f,\l)_D,2)$ & $f2.3:(\l,f^2,\l)_D,1)$ & &  \\
$f1.4: (1,(\l,B,\l)_D,4)$ &&& \\
$f1.5: (1,(\l,f^3,\l)_D,3)$&&&\\
$f1.6: (1,(\l,f''',\l)_D,3)$ &&&\\
$f1.7: (1,(\l,f^{4},\l)_D,2)$ & & & \\
\hline
\end{tabular}}\caption{Star-controlled GCID of size $(4;2,1,1;1,0,0)$ simulating rules of SGNF.}
\label{table-4211100}
\end{center} 
\end{table}}
\noindent
\begin{proof} Consider a type-$0$ grammar $G=(N,T,P,S)$ in SGNF as in Definition~\ref{def-SGNF}. \LV{The rules of $P$ are\LV{ assumed to be} labelled bijectively with labels from $[1 \ldots |P|]$. }We construct a\LV{ star-controlled} GCID system $\Pi= (4,V,T,\{S\},H,1,\{1\},R)$ of size $(4;2,1,1;\linebreak[2]1,0,0)$\LV{ as follows} such that $L(\Pi) = L(G)$. The set $V$ contains the symbols of $G$ as well as some rule markers. We refer to Table~\SV{\ref{table-4211100}}\LV{\ref{table-CF-4211100}} for the direct simulation\LV{ of context-free rules\todo{HF: By small modifications of these rules, it would be possible to reduce the size to  3210100 or to 3201100; this might help stitching new results.}} of SGNF\LV{ and further refer to Table~\ref{table-nonCF-4211100-revised4} for the direct simulation of non-context-free rules}.\LV{ The control graph is shown in Fig.~\ref{fig4star}.}
The  rules simulating $g: CD \to 
\l$ are similar to the ones simulating the $f: AB \to \l$ rule and hence omitted. 

We now prove that $L(G) \subseteq L(\Pi)$. We show that
if $w\Ra w'$ in~$G$, then $(w)_1\Ra_* (w')_1$ according to $\Pi$. From this fact, the claim follows by a simple induction, split into different cases, as discussed now. 

\smallskip
\noindent\underline{Context-free rule $p:X\to bY$.} Here, $w=\a X\b$ and $w'=\a bY\b$ for some $\a\in \{A,C\}^*$, $\b\in (\{B,D\}\cup T)^*$. The simulation performs as follows:
\begin{align*}
(w)_1=(\a X \b)_1  \Ra_{p1.1}    (\a pX\b)_2 \Ra_{p2.1} (\a p \b)_1 \Ra_{p1.2} (\a pbY \b)_3 \Ra_{p3.1} (\a bY\b)_1 =w'.
\end{align*}

\smallskip
\noindent\underline{Context-free rule $q:X\to Yb$} is simulated in a similar fashion to the simulation of the $p$-rule\LV{ described above  and hence the detailed description is omitted here}. 

\smallskip
\noindent\underline{Context-free rule $h:S'\to \emptyword$} is simulated as follows for some $\a\in \{A,C\}^*$, $\b\in (\{B,D\}\cup T)^*$:
$$(w)_1=(\a S'\b)_1 \Ra_{h1.1} (\a hh' S'\b)_2\Ra_{h2.1} (\a hh' \b)_1 \Ra_{h1.2} (\a h\b)_3 \Ra_{h3.1} (\a\b)_1=w'.$$

\smallskip
\noindent\underline{Non-context-free rule $f:AB\ra\l$.}
This means that we expect $w=\a AB\b$ and $w'=\a\b$ for some $\a\in \{A,C\}^*$, $\b\in (\{B,D\}\cup T)^*$. 
Within $\Pi$, this can be simulated as follows.
\begin{center}
\scalebox{.9}{
\begin{minipage}{\textwidth}
\begin{align*} 
(\a AB \b)_1 &\Ra_{f1.1} (\a f'AB \b)_2 \Ra_{f2.1} (\a f'AfB\b)_1 \Ra_{f1.2} (\a f'fB\b)_4 \Ra_{f4.1} (\a f'f''f^2fB\b)_1  \Ra_{f1.3} (\a f'f''f^2B\b)_{2} \\ & \Ra_{f2.2}  (\a f'f''f^2Bf^4 \b)_1  \Ra_{f1.4} (\a f'f''f^2f^4\b)_4 \Ra_{f4.2}  (\a f'f''f^2f'''f^3f^4 \b)_1  \Ra_{f1.5} (\a f'f''f^2f'''f^4 \b)_3 \\ & \Ra_{f3.1} (\a f'f^2f'''f^4 \b)_1 \Ra_{f1.6}  (\a f'f^2f^4 \b)_3 \Ra_{f3.2} (\a f^2f^4\b)_1\Ra_{f1.7}(\a f^2\b)_2\Ra_{f2.3}(\a \b)_1.
\end{align*}
\end{minipage}
}
\end{center}

\smallskip

We next prove that $L(\Pi) \subseteq L(G)$. More precisely, we show the following (by induction).
If $(u)_1$ is a configuration that is derivable in~$\Pi$ such that~$u$ contains the same number of markers from $\{f',g'\}$ as from $\{f^4,g^4\}$ (we call this property of the symbols from $\{f',g',f^4,g^4\}$ also a balanced situation) and such that the word~$u'$ that is obtained from~$u$ by deleting all symbols from  $\{f',g',f^4,g^4\}$ belongs to $(N\cup T)^*$, then~$u$ is derivable in~$G$.
In particular, by induction we can assume that any such~$u$ that we discussed for proving the inductive step satisfies $u'\in \{A,C\}^*\linebreak[3](N'\cup\{\emptyword\})(\{B,D\}\cup T)^*$. To avoid clumsy formulations, we will discuss the markers from  $\{f',g',f^4,g^4\}$ only in particular situations and argue why we maintain the property of being in a balanced situation.\LV{

}
Hence, consider a derivation $(w)_1\Ra^*(w')_1$, with $w\neq w'$ and $w,w'\in \{A,C\}^*\linebreak[3](N'\cup\{\emptyword\})(\{B,D\}\cup T)^*$; it  has to start like $(w)_1\Ra (v)_j$ in~$\Pi$. If some rule from $C1$ is applied to $w$, the rule will insert a rule marker into the string~$w$ and move to $C2$, or an~$A$ or~$B$ is deleted and the string moves to~$C4$. The introduction of rule markers in $C1$ will take care of the non-interference among the non-context-free and context-free rules. We will now discuss \LV{the possibilities to some }more details. 

Let us first assume (by induction) that the sentential form $w^1 = \a X\b$ for some $\a\in \{A,C\}^*$, $\b\in (\{B,D\}\cup T)^*$ and $X \in N'$ is derivable in~$G$ and the configuration $(w^1)_1$ is derivable in~$\Pi$. We \LV{will }prove (as induction step) that if $(w^1)_1\Ra_*(u)_1$, $w^1\neq u$, and $u\in \{A,C\}^*(N'\cup\{\lambda\})(\{B,D\}\cup T)^*$ is the first sentential form in $\{A,C\}^*(N'\cup\{\lambda\})(\{B,D\}\cup T)^*$ that appears in a $\Pi$-derivation\LV{, starting and ending in $C1$}, then $w^1\Ra u$\LV{ holds} in~$G$. \LV{The simulation of the $q$-rules is similar to that of the $p$-rules and hence, we discuss only the correctness of the $p$-rule simulation here.}

\smallskip\noindent
\underline{\it Caveat:} The reader might wonder why a context-free deletion of~$A$ or~$B$ as in rules $f1.2$ or $f1.4$ could work at all. But notice that in $C4$, which is the target component of these rules, the presence of $f$-style markers is checked in each of the rules. This prevents any successful derivation that interferes with, say, a $p$-rule derivation by deleting an $A$ or a $B$ in the first component, simply because there are no $p$-rule simulation rules in $C4$. We will hence tacitly assume that $f1.2$ etc. are not applied within $p$-simulations.

\smallskip\noindent
\underline{Applying $p1.1$} to $w^1 = \a X\b$ for some $\a\in \{A,C\}^*$, $\b\in (\{B,D\}\cup T)^*$ and as $X \in N'$ being the only nonterminal of $N'$ that is in $w^1$, the marker $p$ will be inserted to the left of $X$, yielding $(w^2)_2 =(\a pX\b)_2$. In $C2$, the only available and applicable rule among the rules meant to simulate context-free rules deletes the nonterminal~$X$ (which is not $S'$) of~$N'$. This results in $(w^3)_1= (\a p \b)_1$. Alternatively, rule $f2.2$ (or $g2.2$) would be applicable in $C2$, bringing the string back to $C1$. There, one could return to $C2$ by using $f1.1$, $g1.1$, $f1.7$ or $g1.7$. The last two possibilities (if applicable at all) would just undo the previous step and hence offer no progress in the derivation, or they would exchange an occurrence of $g^4$ with an occurrence of $f^4$ or vice versa. The first two mentioned rules would add a symbol from $\{f',g'\}$ to the string obtained from $w^3$ by previously adding a  symbol from $\{f^4,g^4\}$. Hence, if we think of the statement that our inductive argument should prove, we keep up a balanced situation as required.
Now, in $C1$, the rule $p1.1$ cannot be applied again, as no $X \in N'$ is present. So, the only applicable rule (within the rules in $C1$ that are meant to deal with simulating context-free rules) is $p1.2$ which results in $(w^3)_3=(\a pbY \b)_3$. In $C3$, 
the introduced marker $p$ is deleted and the resultant string $w^4 = \a bY \b$ is sent to $C1$, thus the intended and desired derivation is correctly  simulated. Notice that we could also apply $f3.2$ or $g3.2$ instead of $p3.1$, which would delete some $f'$ or $g'$ that  might have been introduced earlier. However, this would then create an imbalanced situation, with less occurrences of symbols from $\{f',g'\}$ than from $\{f^4,g^4\}$ within a string in $C1$. As we will see\LV{ when studying the derivation within the simulation of $f$-rules}, such an imbalance can never be resolved, so that such a string will not yield a terminal string. 
There is one possibility after applying $p1.1$ and $p2.1$ that still needs to be discussed. It might be possible to apply $f1.1$ on $w^3=\a p \b$. This is only possible if $\a=\a_1 A \a_2$, so that we arrive at $(u^4)_2=(\a_1 f'A \a_2p \b)_2$.
Now, $f2.2$ or $g2.2$ may be applicable, bringing us back to $C1$. However, this will maintain a balanced situation as claimed. 
Moreover, as the reader can check, this is indeed the only possible continuation, keeping in mind that $AB$ will not occur as a substring in the present string due to SGNF. So, in various ways, along with a simulation of a context-free rule, we may add symbols from $\{f',g',f^4,g^4\}$, but this always happens in a balanced way if it might be fruitful. \LV{The working of the simulation of $q$-rules is very similar to that of $p$-rules and it is left to the reader to verify this.}

\smallskip\noindent
\SV{The correctness of the simulation of the \underline{$h$-rule} is easily seen.}\LV{The transition phase from phase I to phase II is correctly simulated as follows, \underline{starting with $h1.1$}. Let $w'= \a S' \b$ be available in $C1$. As no nonterminal of $N'$ other than $S'$ is available in $w'$, no $p1.i$ or $q1.i$ rules ($i=1,2$) are applicable and the only applicable rule is $h1.1$ (cf. the \emph{Caveat} above) and applying it results to $(w'')_2 = (\a hh'S' \b)_2$. In $C2$,  $h2.1$ is the only applicable rule, so we get $(w''')_1 = (\a hh'\b)_1$.  In $C1$, with $w'''$, the only applicable rule is $h1.2$, resulting into $(w^{iv})_3= (\a h \b)_3$ and the remaining marker $h$ is deleted in $C3$ with the rule $h3.1$, resulting into $(w^v)=(\a \b)_1$, as intended.} Notice that there are again possibilities to introduce or delete symbols from $\{f',g',f^4,g^4\}$ similar as discussed above for simulating $p$-rules.

\smallskip\noindent
\underline{Applying $f1.1$} to $w^1 = \a X\b$ for some $\a\in \{A,C\}^*$, $\b\in (\{B,D\}\cup T)^*$ and  $X \in N'$, the marker $f'$ will be inserted to the left of some occurrence of~$A$, yielding $(w^2)_2 =(\a_1 f'A\a_2 X\b)_2$, with $\a=\a_1A\a_2$. As $w^1 = \a X\b$ indicates that we are simulating phase I of the work of the SGNF grammar~$G$, the substring $AB$ is absent, preventing us from applying $f2.1$. We might continue with $f2.2$ or $g2.2$, though, which introduces an occurrence of $f^4$ or $g^4$, respectively. This is one possibility how we can obtain `pairs' of occurrences of symbols from $\{f',g'\}$ and $\{f^4,g^4\}$, but we clearly maintain a balanced situation.  
Still, we might delete~$X$ with $p2.1$, say, but then we arrive at a typical imbalanced situation. 
We would have to apply $f1.1$ or $f1.2$ or $f1.4$ next, as no markers from the context-free rule simulations are present. In the first case, we see that we maintain an imbalanced situation of which we can never get rid, while in the second and third case, the derivation is blocked in $C4$ because of the lack of appropriate $f$-markers.

We could lead a similar discussion for \underline{applying $f1.7$} to a string that contains one occurrence of $N'$-symbols and at least one occurrence of $f^4$; in particular, the arguments concerning (im)balance remain the same, because for this condition, it does not matter whether we add $f'$\LV{ (by applying $f1.1$)} or\LV{ whether we} delete $f^4$\LV{ (by applying $f1.7$)}.

\smallskip
\noindent We shall discuss the cases for the \underline{$f$-rule simulation} next, including that the rules $f1.i$ were applied in a wrong manner. Recall the discussion of the \emph{Caveat} above which ruled out a premature application of $f1.2$ or of $f1.4$. We cannot start with any rules that require the presence of $f$-markers, apart from those stemming from a balanced situation that we will discuss below. \LV{

}We therefore discuss a derivation starting with $f1.1$ on $w^1$ as intended. Again by the absence of the marker $f''$, we have to apply $f2.1$ next, shifting the discussion of a balanced situation as created after applying $f2.2$ to what we say further down. The role of rule $f2.1$ is crucial insofar as it checks that the substring $AB$ is present in the current string. It is one of the important properties of SGNF that this substring can only occur once in a derived string, and this also means that we are in phase II of the SGNF derivation. 
By induction, we can assume this property also to hold for the string $w^1$ that is under discussion. In other words, we can assume that $w^1=\a AB\b$. Then, after applying $f1.1$ and $f2.1$, we are in the configuration $(w^2)_1=(\a'f'\a''AfB\b)_1$ with $\a=\a'\a''$, and $\a''$ being a string that is either empty or it starts with an~$A$. The intention would be to apply $f1.2$ on $w^2$. As the marker $f^2$ is absent, we have to apply $f4.1$ now. This is only possible if $\a''$ is empty.
Hence, $(w^2)_1=(\a f'AfB\b)_1$, and after applying $f1.2$ and $f4.1$, we necessarily arrive at the configuration $(w^3)_1=(\a f'f''f^2fB\b)_1$.
We now check the other possibilities in configuration $(w^2)_1$. Due to the absence of $f^3$, $f'''$ and $f^4$, only $f1.1$, $f1.3$, or $f1.4$ are applicable. If we apply $f1.1$, then  as there is no second substring $AB$ nor the $f$-marker $f^2$  present, we have to introduce $f^4$ next to obtain a configuration $(u)_1$ where $u$ contains two occurrences of~$f'$ and one occurrence of $f$ and one of~$f^4$.\LV{ Applying $f1.3$ or $f1.7$ on~$u$ clearly gives no progress.} 
So, on $(u)_1$, we might apply $f1.2$, followed by $f4.1$.
As argued above for the main line of derivation, this means that we arrive at a configuration $(v)_1$ with $v=\a_1f'\a_2f'f''f^2f\b_1Bf^4\b_2$, with $\a_1\a_2=\a$ and $\b_1B\b_2=\b$. In fact, we could continue now with the derivation, closely following the main line, the only difference being an additional $f'$ and $f^4$ being present in the string. This indicates that balanced situations are not necessarily a problem and also shows how they can arise.
On applying $f1.3$, we can delete $f$, but this takes us back to\LV{ the configuration} $(w^2)_2$\LV{ that we saw before}.
This analysis tells us that, in configuration  $(w^2)_1$, we have to apply $f1.2$.

Now, we study the configuration $(w^3)_1=(\a f'f''f^2fB\b)_1$. If we apply $f1.1$, we again have to introduce an occurrence of $f^4$ and can then follow the main line of derivation, which will finally lead to a balanced situation. If we delete $A$ or~$B$, the derivation is stuck in $C4$. Hence, the only applicable, promising rule is $f1.3$, deleting $f$ and moving to $C2$. After deleting~$f^2$, we get back to $C1$ with nearly the same (im)possibilities as just discussed, except that $f1.3$ is no longer available. Therefore, such a derivation cannot lead to a terminal string. Hence, in $C2$ the rule $f2.2$ must be applied, giving, with $\b'\b''=\b$ and $\b'$ being empty or ending with~$B$,  $(w^3)_1=(\a f'f''f^2fB\b'f^4\b'')_1$. Again, we might add a second~$f'$ with no fruitful continuation now, except for possibly creating other balanced situations finally. If we delete~$A$, then the derivation is stuck in $C4$. If we delete $f^4$ with $f1.7$, we have to delete~$f^2$ in $C2$ (unless we want to un-do $f1.7$ by applying $f2.2$)  and then, whatever we apply next ($f1.1$ or $f1.2$ or $f1.4$), the derivation is stuck, ignoring the possibility to add more and more occurrences of~$f'$ and~$f^4$. Therefore, on~$w^3$, we have to apply $f1.4$ which deletes one~$B$. In $C4$, we can only apply any rule if $\b'$ is  empty, so that $f4.2$ is applicable, resulting in $(w^4)_1=(\a f'f''f^2f'''f^3f^4\b)_1$.
The next six rule applications will delete all six $f$-markers; several ways to do this are possible.
What could go wrong?
Introducing a second $f'$-occurrence is again not interesting: it has to be matched by adding a $f^4$-occurrence; if we delete $f^2$ instead, then the addition of $f^4$ is indirect insofar, as $f1.7$ need not be executed to delete $f^2$.
After applying $f1.5$, $f3.1$, $f1.6$ and $f3.2$, we end up with a string from which one could delete any $A$- or $B$-occurrence and try to restart with $f4.2$. Yet, now the symbols~$f'''$ and~$f^3$ can only be deleted if one introduces two additional occurrences of $\{f',f''\}$, requiring a complete re-start of the simulation; it is not possible to get rid of~$f'''$ and~$f^3$ otherwise. Such a re-start would require the substring $AB$ which is currently not present. Hence, $(\a f^2f^4\b)_1$ can only be continued as intended, applying $f1.7$ and $f2.3$.

\smallskip
\noindent \underline{More on balanced situations.}
As we have been mentioning these over and over again in the previous arguments, let us briefly discuss possibilities when we do have some balanced occurrences from $\{f',g',f^4,g^4\}$ in configuration $(w^1)_1$, e.g., $w^1=\a AB \b$ where $\a\in\{A,C\}^*\{f'\}$  and $\b$ contains exactly one occurrence from $\{f^4,g^4\}$. Then, instead of applying $f1.1$, we could start the derivation with $f1.7$ or $g1.7$. Notice that this yields a configuration $(w^2)_2$ that could have also been obtained when starting from $(\a'AB\b')_1$ and applying $f1.1$, where $\a'$ is obtained from $\a$ by deleting $f'$ and $\b'$ is obtained from $\b$ by deleting the unique occurrence of either $f^4$ or $g^4$. This shows that balanced situations can lead to terminal strings finally, as they may converge again to the main line of derivation. Importantly, no\LV{ substantially new derivations can be obtained, so that in particular no} new terminal strings can be derived that are not following the possibilities given by~$G$.

This concludes the main arguments concerning the inductive step and hence the claim follows.
\end{proof}

\section{Summary and Open Problems}

In this paper, we focused on examining the computational power of graph-controlled ins-del systems with
a star as a control graph. We lowered the resource requirements to describe $\mathrm{RE}$, all recursively enumerable languages. 
We leave it open to explore the following possibilities. 
\begin{enumerate}

\item $\mathrm{GCID}_S(k;2,i',i'';1,j',j'') \stackrel{?}{=} \mathrm{RE}$ for $i'+i''\leq 1$ and $ j'+j'' \leq  1$ and some $k$ as small as possible, 

\item $\mathrm{GCID}_S(k';2,0,0;1,i',i'') \stackrel{?}{=} \mathrm{RE}$ for $i'+i''\leq 1$ and some $k'$ as small as possible. 

\end{enumerate}

Here we  only considered $\mathrm{GCID}$ systems where the underlying graph is star-controlled and does not contain loops. One may also consider a tree structure and / or the possibility to allow loops (i.e., rules have option `here' and the resultant string can stay back in the same component if such rules are applied), which may give additional power and connect closer to ins-del P systems and also to the results of~\cite{IvaVer2017}.

{
\bibliographystyle{eptcs}

\bibliography{AFL-FKR}

\begin{thebibliography}{10}
\providecommand{\bibitemdeclare}[2]{}
\providecommand{\surnamestart}{}
\providecommand{\surnameend}{}
\providecommand{\urlprefix}{Available at }
\providecommand{\url}[1]{\texttt{#1}}
\providecommand{\href}[2]{\texttt{#2}}
\providecommand{\urlalt}[2]{\href{#1}{#2}}
\providecommand{\doi}[1]{doi:\urlalt{https://doi.org/#1}{#1}}
\providecommand{\eprint}[1]{arXiv:\urlalt{https://arxiv.org/abs/#1}{#1}}
\providecommand{\bibinfo}[2]{#2}

\bibitemdeclare{article}{AlhFIV2022}
\bibitem{AlhFIV2022}
\bibinfo{author}{A.~\surnamestart Alhazov\surnameend},
  \bibinfo{author}{R.~\surnamestart Freund\surnameend},
  \bibinfo{author}{S.~\surnamestart Ivanov\surnameend} \&
  \bibinfo{author}{S.~\surnamestart Verlan\surnameend} (\bibinfo{year}{2022}):
  \emph{\bibinfo{title}{Regulated Insertion-Deletion Systems}}.
\newblock {\slshape \bibinfo{journal}{Journal of Automata, Languages and
  Combinatorics}} \bibinfo{volume}{27}(\bibinfo{number}{1-3}), pp.
  \bibinfo{pages}{15--45}, \doi{10.25596/jalc-2022-015}.

\bibitemdeclare{book}{Ben93}
\bibitem{Ben93}
\bibinfo{editor}{R.~\surnamestart Benne\surnameend}, editor
  (\bibinfo{year}{1993}): \emph{\bibinfo{title}{{RNA} Editing: The Alteration
  of Protein Coding Sequences of {RNA}}}.
\newblock \bibinfo{publisher}{Ellis Horwood}.

\bibitemdeclare{article}{FerKupRam2017b}
\bibitem{FerKupRam2017b}
\bibinfo{author}{H.~\surnamestart Fernau\surnameend},
  \bibinfo{author}{L.~\surnamestart Kuppusamy\surnameend} \&
  \bibinfo{author}{I.~\surnamestart Raman\surnameend} (\bibinfo{year}{2017}):
  \emph{\bibinfo{title}{On the computational completeness of graph-controlled
  insertion-deletion systems with binary sizes}}.
\newblock {\slshape \bibinfo{journal}{Theoretical Computer Science}}
  \bibinfo{volume}{682}, pp. \bibinfo{pages}{100--121},
  \doi{10.1016/j.tcs.2017.01.019}.

\bibitemdeclare{article}{FerKupRam2019}
\bibitem{FerKupRam2019}
\bibinfo{author}{H.~\surnamestart Fernau\surnameend},
  \bibinfo{author}{L.~\surnamestart Kuppusamy\surnameend} \&
  \bibinfo{author}{I.~\surnamestart Raman\surnameend} (\bibinfo{year}{2019}):
  \emph{\bibinfo{title}{On path-controlled insertion-deletion systems}}.
\newblock {\slshape \bibinfo{journal}{Acta Informatica}}
  \bibinfo{volume}{56}(\bibinfo{number}{1}), pp. \bibinfo{pages}{35--59},
  \doi{10.1007/s00236-018-0312-2}.

\bibitemdeclare{article}{FerKupRam2021a}
\bibitem{FerKupRam2021a}
\bibinfo{author}{H.~\surnamestart Fernau\surnameend},
  \bibinfo{author}{L.~\surnamestart Kuppusamy\surnameend} \&
  \bibinfo{author}{I.~\surnamestart Raman\surnameend} (\bibinfo{year}{2021}):
  \emph{\bibinfo{title}{On the generative capacity of matrix insertion-deletion
  systems of small sum-norm}}.
\newblock {\slshape \bibinfo{journal}{Natural Computing}}
  \bibinfo{volume}{20}(\bibinfo{number}{4}), pp. \bibinfo{pages}{671--689},
  \doi{10.1007/s11047-021-09866-y}.

\bibitemdeclare{inproceedings}{FreKRV2010}
\bibitem{FreKRV2010}
\bibinfo{author}{R.~\surnamestart Freund\surnameend},
  \bibinfo{author}{M.~\surnamestart Kogler\surnameend},
  \bibinfo{author}{Y.~\surnamestart Rogozhin\surnameend} \&
  \bibinfo{author}{S.~\surnamestart Verlan\surnameend} (\bibinfo{year}{2010}):
  \emph{\bibinfo{title}{Graph-Controlled Insertion-Deletion Systems}}.
\newblock In \bibinfo{editor}{I.~\surnamestart McQuillan\surnameend} \&
  \bibinfo{editor}{G.~\surnamestart Pighizzini\surnameend}, editors: {\slshape
  \bibinfo{booktitle}{Proceedings Twelfth Annual Workshop on Descriptional
  Complexity of Formal Systems, {DCFS}}}, {\slshape
  \bibinfo{series}{{EPTCS}}}~\bibinfo{volume}{31}, pp. \bibinfo{pages}{88--98},
  \doi{10.4204/EPTCS.31.11}.

\bibitemdeclare{article}{Gef91a}
\bibitem{Gef91a}
\bibinfo{author}{V.~\surnamestart Geffert\surnameend} (\bibinfo{year}{1991}):
  \emph{\bibinfo{title}{Normal forms for phrase-structure grammars}}.
\newblock {\slshape \bibinfo{journal}{{RAIRO} Informatique th\'eorique et
  Applications/Theoretical Informatics and Applications}} \bibinfo{volume}{25},
  pp. \bibinfo{pages}{473--498}, \doi{10.1051/ita/1991250504731}.

\bibitemdeclare{inproceedings}{IvaVer1314}
\bibitem{IvaVer1314}
\bibinfo{author}{S.~\surnamestart Ivanov\surnameend} \&
  \bibinfo{author}{S.~\surnamestart Verlan\surnameend} (\bibinfo{year}{2014}):
  \emph{\bibinfo{title}{About One-Sided One-Symbol Insertion-Deletion {P}
  Systems}}.
\newblock In \bibinfo{editor}{A.~\surnamestart Alhazov\surnameend},
  \bibinfo{editor}{S.~\surnamestart Cojocaru\surnameend},
  \bibinfo{editor}{M.~\surnamestart Gheorghe\surnameend},
  \bibinfo{editor}{Y.~\surnamestart Rogozhin\surnameend},
  \bibinfo{editor}{G.~\surnamestart Rozenberg\surnameend} \&
  \bibinfo{editor}{A.~\surnamestart Salomaa\surnameend}, editors: {\slshape
  \bibinfo{booktitle}{Membrane Computing - 14th Int. Conf., {CMC} 2013}},
  {\slshape \bibinfo{series}{LNCS}} \bibinfo{volume}{8340},
  \bibinfo{publisher}{Springer}, pp. \bibinfo{pages}{225--237},
  \doi{10.1007/978-3-642-54239-8\_16}.

\bibitemdeclare{article}{IvaVer2017}
\bibitem{IvaVer2017}
\bibinfo{author}{S.~\surnamestart Ivanov\surnameend} \&
  \bibinfo{author}{S.~\surnamestart Verlan\surnameend} (\bibinfo{year}{2017}):
  \emph{\bibinfo{title}{Universality and Computational Completeness of
  Controlled Leftist Insertion-Deletion Systems}}.
\newblock {\slshape \bibinfo{journal}{Fundamenta Informaticae}}
  \bibinfo{volume}{155}(\bibinfo{number}{1-2}), pp. \bibinfo{pages}{163--185},
  \doi{10.3233/FI-2017-1580}.

\bibitemdeclare{incollection}{KarPTY99}
\bibitem{KarPTY99}
\bibinfo{author}{L.~\surnamestart Kari\surnameend}, \bibinfo{author}{{\relax
  Gh}.~\surnamestart P\u{a}un\surnameend}, \bibinfo{author}{G.~\surnamestart
  Thierrin\surnameend} \& \bibinfo{author}{S.~\surnamestart Yu\surnameend}
  (\bibinfo{year}{1999}): \emph{\bibinfo{title}{At the crossroads of {DNA}
  computing and formal languages: Characterizing recursively enumerable
  languages using insertion-deletion systems}}.
\newblock In: {\slshape \bibinfo{booktitle}{Discrete Mathematics and Theretical
  Computer Science}}, {\slshape \bibinfo{series}{DIMACS}}~\bibinfo{volume}{48},
  \bibinfo{publisher}{AMS}, pp. \bibinfo{pages}{329--338},
  \doi{10.1090/dimacs/048/23}.

\bibitemdeclare{article}{KarThi96}
\bibitem{KarThi96}
\bibinfo{author}{L.~\surnamestart Kari\surnameend} \&
  \bibinfo{author}{G.~\surnamestart Thierrin\surnameend}
  (\bibinfo{year}{1996}): \emph{\bibinfo{title}{Contextual Insertions/Deletions
  and Computability}}.
\newblock {\slshape \bibinfo{journal}{Information and Computation}}
  \bibinfo{volume}{131}(\bibinfo{number}{1}), pp. \bibinfo{pages}{47--61},
  \doi{10.1006/inco.1996.0091}.

\bibitemdeclare{inproceedings}{KraRogVer2008}
\bibitem{KraRogVer2008}
\bibinfo{author}{A.~\surnamestart Krassovitskiy\surnameend},
  \bibinfo{author}{Y.~\surnamestart Rogozhin\surnameend} \&
  \bibinfo{author}{S.~\surnamestart Verlan\surnameend} (\bibinfo{year}{2008}):
  \emph{\bibinfo{title}{Further Results on Insertion-Deletion Systems with
  One-Sided Contexts}}.
\newblock In \bibinfo{editor}{C.~\surnamestart
  Mart{\'{\i}}n{-}Vide\surnameend}, \bibinfo{editor}{F.~\surnamestart
  Otto\surnameend} \& \bibinfo{editor}{H.~\surnamestart Fernau\surnameend},
  editors: {\slshape \bibinfo{booktitle}{Language \& Automata Theory \&
  Applications, {LATA}}}, {\slshape \bibinfo{series}{LNCS}}
  \bibinfo{volume}{5196}, \bibinfo{publisher}{Springer}, pp.
  \bibinfo{pages}{333--344}, \doi{10.1007/978-3-540-88282-4\_31}.

\bibitemdeclare{inproceedings}{KriRam0102}
\bibitem{KriRam0102}
\bibinfo{author}{S.~N. \surnamestart Krishna\surnameend} \&
  \bibinfo{author}{R.~\surnamestart Rama\surnameend} (\bibinfo{year}{2002}):
  \emph{\bibinfo{title}{Insertion-Deletion {P} Systems}}.
\newblock In \bibinfo{editor}{N.~\surnamestart Jonoska\surnameend} \&
  \bibinfo{editor}{N.~C. \surnamestart Seeman\surnameend}, editors: {\slshape
  \bibinfo{booktitle}{{DNA} Computing, 7th Int. Workshop on DNA-Based
  Computers, 2001}}, {\slshape \bibinfo{series}{LNCS}} \bibinfo{volume}{2340},
  \bibinfo{publisher}{Springer}, pp. \bibinfo{pages}{360--370},
  \doi{10.1007/3-540-48017-X\_34}.

\bibitemdeclare{article}{Mar69}
\bibitem{Mar69}
\bibinfo{author}{S.~\surnamestart Marcus\surnameend} (\bibinfo{year}{1969}):
  \emph{\bibinfo{title}{Contextual grammars}}.
\newblock {\slshape \bibinfo{journal}{Revue Roumaine de Math\'ematiques Pures
  et Appliqu\'ees}} \bibinfo{volume}{14}, pp. \bibinfo{pages}{1525--1534}.

\bibitemdeclare{inproceedings}{MatRogVer2007}
\bibitem{MatRogVer2007}
\bibinfo{author}{A.~\surnamestart Matveevici\surnameend},
  \bibinfo{author}{Y.~\surnamestart Rogozhin\surnameend} \&
  \bibinfo{author}{S.~\surnamestart Verlan\surnameend} (\bibinfo{year}{2007}):
  \emph{\bibinfo{title}{Insertion-Deletion Systems with One-Sided Contexts}}.
\newblock In: {\slshape \bibinfo{booktitle}{MCU}}, {\slshape
  \bibinfo{series}{LNCS}} \bibinfo{volume}{4664},
  \bibinfo{publisher}{Springer}, pp. \bibinfo{pages}{205--217},
  \doi{10.1007/978-3-540-74593-8\_18}.

\bibitemdeclare{book}{Pau97}
\bibitem{Pau97}
\bibinfo{author}{{\relax Gh}.~\surnamestart P{\u a}un\surnameend}
  (\bibinfo{year}{1997}): \emph{\bibinfo{title}{Marcus Contextual Grammars}}.
\newblock {\slshape \bibinfo{series}{Studies in Linguistics and
  Philosophy}}~\bibinfo{volume}{67}, \bibinfo{publisher}{Kluwer},
  \doi{10.1007/978-94-015-8969-7_4}.

\bibitemdeclare{book}{PauRozSal98a}
\bibitem{PauRozSal98a}
\bibinfo{author}{{\relax Gh}.~\surnamestart P{\u a}un\surnameend},
  \bibinfo{author}{G.~\surnamestart Rozenberg\surnameend} \&
  \bibinfo{author}{A.~\surnamestart Salomaa\surnameend} (\bibinfo{year}{1998}):
  \emph{\bibinfo{title}{DNA Computing: New Computing Paradigms}}.
\newblock \bibinfo{publisher}{Springer}, \doi{10.1007/978-3-662-03563-4}.

\bibitemdeclare{article}{PetVer2012}
\bibitem{PetVer2012}
\bibinfo{author}{I.~\surnamestart Petre\surnameend} \&
  \bibinfo{author}{S.~\surnamestart Verlan\surnameend} (\bibinfo{year}{2012}):
  \emph{\bibinfo{title}{Matrix insertion-deletion systems}}.
\newblock {\slshape \bibinfo{journal}{Theoretical Computer Science}}
  \bibinfo{volume}{456}, pp. \bibinfo{pages}{80--88},
  \doi{10.1016/j.tcs.2012.07.002}.

\bibitemdeclare{article}{TakYok2003}
\bibitem{TakYok2003}
\bibinfo{author}{A.~\surnamestart Takahara\surnameend} \&
  \bibinfo{author}{T.~\surnamestart Yokomori\surnameend}
  (\bibinfo{year}{2003}): \emph{\bibinfo{title}{On the computational power of
  insertion-deletion systems}}.
\newblock {\slshape \bibinfo{journal}{Natural Computing}}
  \bibinfo{volume}{2}(\bibinfo{number}{4}), pp. \bibinfo{pages}{321--336},
  \doi{10.1023/B:NACO.0000006769.27984.23}.

\bibitemdeclare{article}{Ver2010}
\bibitem{Ver2010}
\bibinfo{author}{S.~\surnamestart Verlan\surnameend} (\bibinfo{year}{2010}):
  \emph{\bibinfo{title}{Recent Developments on Insertion-Deletion Systems}}.
\newblock {\slshape \bibinfo{journal}{The Computer Science Journal of Moldova}}
  \bibinfo{volume}{18}(\bibinfo{number}{2}), pp. \bibinfo{pages}{210--245}.

\end{thebibliography}



\end{document}